\pdfoutput=1
\RequirePackage{ifpdf}
\ifpdf
\documentclass[pdftex]{sigma}
\else
\documentclass{sigma}
\fi

\usepackage{mathtools}

\numberwithin{equation}{section}

\newtheorem{Theorem}{Theorem}[section]
\newtheorem{Lemma}[Theorem]{Lemma}
\newtheorem{Proposition}[Theorem]{Proposition}

{\theoremstyle{definition}
\newtheorem{Remark}[Theorem]{Remark}
\newtheorem*{ass*}{Assumptions}
\newtheorem*{rh-pb*}{Main RH problem}
\newtheorem*{properties*}{Properties}
\newtheorem*{fact-pb*}{Factorization problem}
\newtheorem*{add-conds*}{Additional conditions}}

\newcommand{\D}[1]{\mathbb{#1}}
\newcommand{\dd}{\mathrm{d}}
\newcommand{\eul}{\mathrm{e}}
\newcommand{\ii}{\mathrm{i}}

\providecommand{\accol}[1]{\lbrace#1\rbrace}
\providecommand{\croch}[1]{\lbrack#1\rbrack}

\newcommand{\DP}{\mathrm{DP}}

\newcommand{\ord}{\mathrm{O}}
\newcommand{\osmall}{\mathrm{o}}
\renewcommand{\Re}{\operatorname{Re}}
\DeclareMathOperator{\diag}{diag}
\DeclareMathOperator{\Res}{Res}

\begin{document}


\renewcommand{\thefootnote}{$\star$}

\newcommand{\arXivNumber}{1603.08842}

\renewcommand{\PaperNumber}{095}

\FirstPageHeading

\ShortArticleName{A Riemann--Hilbert Approach for the Novikov Equation}
\ArticleName{A Riemann--Hilbert Approach\\ for the Novikov Equation\footnote{This paper is a~contribution to the Special Issue on Asymptotics and Universality in Random Matrices, Random Growth Processes, Integrable Systems and Statistical Physics in honor of Percy Deift and Craig Tracy. The full collection is available at \href{http://www.emis.de/journals/SIGMA/Deift-Tracy.html}{http://www.emis.de/journals/SIGMA/Deift-Tracy.html}}}

\Author{Anne BOUTET DE MONVEL~$^{\dag}$, Dmitry SHEPELSKY~$^\ddag$ and Lech ZIELINSKI~$^\S$}
\AuthorNameForHeading{A.~Boutet de Monvel, D.~Shepelsky and L.~Zielinski}

\Address{$^\dag$~Institut de Math\'ematiques de Jussieu-PRG, Universit\'e Paris Diderot, \\
\hphantom{$^\dag$}~75205 Paris Cedex 13, France}
\EmailD{\href{mailto:anne.boutet-de-monvel@imj-prg.fr}{anne.boutet-de-monvel@imj-prg.fr}}
\URLaddressD{\url{https://webusers.imj-prg.fr/~anne.boutet-de-monvel/}}

\Address{$^\ddag$~Mathematical Division, Institute for Low Temperature Physics,\\
\hphantom{$^\ddag$}~47 Nauki Avenue, 61103 Kharkiv, Ukraine}
\EmailD{\href{mailto:shepelsky@yahoo.com}{shepelsky@yahoo.com}}

\Address{$^\S$~LMPA, Universit\'e du Littoral C\^ote d'Opale,\\
\hphantom{$^\S$}~50 rue F.~Buisson, CS 80699, 62228 Calais, France}
\EmailD{\href{mailto:Lech.Zielinski@lmpa.univ-littoral.fr}{Lech.Zielinski@lmpa.univ-littoral.fr}}

\ArticleDates{Received June 08, 2016, in f\/inal form September 14, 2016; Published online September 24, 2016}

\vspace{-1mm}

\Abstract{We develop the inverse scattering transform method for the Novikov equation $u_t-u_{txx}+4u^2u_x=3u u_xu_{xx}+u^2u_{xxx}$ considered on the line $x\in(-\infty,\infty)$ in the case of non-zero constant background. The approach is based on the analysis of an associated Riemann--Hilbert (RH) problem, which in this case is a $3\times 3$ matrix problem. The structure of this RH problem shares many common features with the case of the Degasperis--Procesi (DP)  equation having quadratic nonlinear terms (see [Boutet~de Monvel A., Shepelsky D., \textit{Nonlinearity} \textbf{26} (2013), 2081--2107, arXiv:1107.5995]) and thus the Novikov equation can be viewed as a~``modif\/ied DP equation'', in analogy with the relationship between the Korteweg--de Vries (KdV) equation and the modif\/ied Korteweg--de Vries (mKdV) equation. We present parametric formulas giving the solution of the Cauchy problem for the Novikov equation in terms of the solution of the RH problem and discuss the possibilities to use the developed formalism for further studying of the Novikov equation.}


\Classification{35Q53; 37K15; 35Q15; 35B40; 35Q51; 37K40}

\vspace{-1mm}

\begin{flushright}
\begin{minipage}{70mm}
\it Dedicated to Percy Deift and Craig Tracy\\ on the occasion of their 70th birthdays.
\end{minipage}
\end{flushright}

\renewcommand{\thefootnote}{\arabic{footnote}}
\setcounter{footnote}{0}

\vspace{-4.5mm}

\section{Introduction} \label{sec:intro}

{\bf 1.1.}~In this paper we present an inverse scattering approach, based on an appropriate Riemann--Hilbert problem formulation, for the initial value problem for the Novikov equation \cite{HW08,MN02,N}
\begin{alignat}{3} \label{N1}
&u_t-u_{txx}+4u^2u_x=3u u_xu_{xx}+u^2u_{xxx},\qquad&&-\infty<x<+\infty,\quad t>0,&\\
&u(x,0)=u_0(x),&&-\infty<x<+\infty, & \nonumber 
\end{alignat}
where $u_0(x)$ is assumed to decay to a non-zero constant:
\begin{gather*}
u_0(x)\to\varkappa>0, \qquad x\to\pm\infty.
\end{gather*}
The solution $u(x,t)$ is also assumed to decay to $\varkappa$ for all $t>0$:
\begin{gather*}
u(x,t)\to\varkappa,\qquad x\to\pm\infty.
\end{gather*}
Introducing the momentum variable
\begin{gather*}
m\coloneqq u-u_{xx},
\end{gather*}
the Novikov equation \eqref{N1} can be written as
\begin{gather}\label{N2-0}
m_t+ (m_xu+3mu_x)u=0
\end{gather}
or, equivalently,
\begin{gather}\label{N2}
\big(m^{\frac{2}{3}}\big)_t+\big(u^2m^{\frac{2}{3}}\big)_x=0.
\end{gather}

{\bf 1.2.}~The Novikov equation \eqref{N1} was obtained in the search for a classif\/ication of integrable generalized Camassa--Holm equations of the form
\begin{gather*}
\big(1-\partial_x^2\big)u_t=F(u,u_x, u_{xx}, u_{xxx},\dots), \qquad u=u(x,t),
\qquad\partial_x=\partial/\partial x
\end{gather*}
possessing inf\/inite hierarchies of higher symmetries \cite{N} (see also~\cite{MN02}). Equation~\eqref{N1} listed by Novikov in \cite[equation~(31)]{N} attracted further considerable attention in the literature, f\/irst of all, as an example of nonlinear equation admitting, like the Camassa--Holm (CH)~\cite{CH93,CHH94} and the Degasperis--Procesi (DP)~\cite{DHH02,DP99} equations, peaked solutions (peakons), but having cubic (rather than quadratic) nonlinear terms. Another integrable Camassa--Holm type equation with cubic nonlinearities was discovered by Fokas~\cite{F} and Qiao \cite{Q06,Q07}.

Hone, Lundmark, and Szmigielski \cite{HLS09} obtained explicit formulas for multipeakon solutions of \eqref{N1}. Some smooth and nonsmooth soliton solutions were presented by Pan and Yi in \cite{PY15}. For studies concerned with blow-up phenomenon and the Cauchy problem for \eqref{N1} we refer the reader to \cite{CCL15, CGLQ16, G13,HH12,JN12, LLW13, NZ11, T11, YLZ12}.

{\bf 1.3.}~In \cite{N} Novikov presented a scalar Lax pair for~\eqref{N1}, which involves the third order derivative with respect to $x$. Hone and Wang~\cite{HW08} proposed a $3\times 3$ matrix Lax pair for~\eqref{N1}, which allowed presenting explicit formulas for peakon solutions~\cite{HLS09,HW08} on zero background. Recently, Matsuno~\cite{M13} presented parametric representations of smooth multisoliton solutions (as well as singular solitons with single cusp and double peaks) of \eqref{N1} on a constant (non-zero) background, using a Hirota-type, purely algebraic procedure. He also demonstrated that a~smooth soliton converges to a~peakon in the limit where the constant background tends to~$0$ while the velocity of the soliton is f\/ixed. Furthermore, he performed the asymptotic analysis of pure multisoliton solutions and noticed that the formulas for the phase shifts of the solitons as well as their peakon limits coincide with those for the Degasperis--Procesi (DP) equation~\cite{M05,M06}
\begin{gather}\label{DP1}
u_t-u_{txx}+4uu_x=3 u_xu_{xx}+uu_{xxx},
\end{gather}
which, in terms of $m$ reads
\begin{gather}\label{DP2}
\big(m^{\frac{1}{3}}\big)_t+\big(um^{\frac{1}{3}}\big)_x=0.
\end{gather}

Comparing with the Degasperis--Procesi equation, it is natural to view (at least formally) the Novikov equation as a~``modif\/ied DP equation'', in analogy with the relationship between the Korteweg--de Vries (KdV) equation
\begin{gather*}
u_t+6uu_x+u_{xxx}=0
\end{gather*}
and the modif\/ied Korteweg--de Vries (mKdV) equation
\begin{gather*}
u_t+6u^2u_x+u_{xxx}=0.
\end{gather*}

The subsequent analysis presented in the paper supports this point of view. Indeed, as we will show below, the implementation of the inverse scattering transform method involves a~Riemann--Hilbert problem of the same structure as in the case of the DP equation (recall that this is true when comparing the KdV and the mKdV equations).

{\bf 1.4.}~Recall that for the DP equation, the transformation $u(x,t)\mapsto \tilde u(x-\varkappa t,t)+\varkappa$ reduces~\eqref{DP1} and~\eqref{DP2} to
\begin{gather}\label{DP3}
\tilde u_t-\tilde u_{txx}+3\varkappa\tilde u_x+4\tilde u\tilde u_x=3 \tilde u_x\tilde u_{xx}+\tilde u\tilde u_{xxx}
\end{gather}
and
\begin{gather*}
\bigl(\tilde m^{\frac{1}{3}}\bigr)_t+\bigl(\tilde u\tilde m^{\frac{1}{3}}\bigr)_x=0,
\end{gather*}
respectively. Here $\tilde m\coloneqq\tilde u-\tilde u_{xx}+\varkappa$ and thus $\tilde u\to 0$ as $x\to\pm\infty$ provided $u\to\varkappa$ as $x\to\pm\infty$. Therefore, the study of the Cauchy problem for the DP equation in the form~\eqref{DP1} in the case of non-zero constant background is equivalent to the study of the Cauchy problem on zero background for the DP equation in the form~\eqref{DP3}, i.e., for the DP equation with \emph{non-zero linear dispersion term}. For the latter problem, the Riemann--Hilbert approach, which is a variant of the inverse scattering transform method, has been applied in \cite{BS13}, which allows obtaining a useful representation of the solution in a form suitable for the analysis of its long time behavior. Notice that the further transformation $\tilde u\mapsto\tilde{\tilde u}\colon\tilde u(x,t)=\varkappa\tilde{\tilde u}(x,\varkappa t)$, which preserves the zero background, allows reducing the study of \eqref{DP3} with any $\varkappa>0$ to the case of~\eqref{DP3} with $\varkappa=1$.

Similar arguments for the Novikov equation lead to the following: The transformations $u(x,t)=\tilde u(x-\varkappa^2t,t) + \varkappa$ and $u(x,t)=\varkappa\tilde {\tilde u}(x-\varkappa^2t,\varkappa^2 t) + \varkappa$ reduce the Cauchy problem for~\eqref{N1} (or~\eqref{N2}) on a non-zero constant background $u\to\varkappa$ as $x\to\pm\infty$ to the Cauchy problem on zero background ($\tilde u\to 0$ and $\tilde{\tilde u}\to 0$ as $x\to\pm\infty$) for the equations
\begin{gather}\label{N3}
\bigl(\tilde m^{\frac{2}{3}}\bigr)_t+\bigl(\big(\tilde u^2+2\varkappa \tilde u\big)\tilde m^{\frac{2}{3}}\bigr)_x=0,
\end{gather}
and
\begin{gather}\label{N4}
\bigl(\tilde{\tilde m}^{\frac{2}{3}}\bigr)_t+\bigl(\big(\tilde{\tilde u}^2+2\tilde{\tilde u}\big)\tilde{\tilde m}^{\frac{2}{3}}\bigr)_x=0,
\end{gather}
respectively. Here $\tilde m\coloneqq\tilde u-\tilde u_{xx}+\varkappa$ and $\tilde{\tilde m}=\tilde{\tilde u}-\tilde{\tilde u}_{xx}+1$.

{\bf 1.5.}~Henceforth, we consider the Cauchy problem for equation~\eqref{N4} on zero background, which, to simplify notations, will be written as
\begin{subequations} \label{N-Cauchy}
\begin{alignat}{3} \label{N5}
&\bigl(\hat{m}^{2/3}\bigr)_t+\bigl(\big(u^2+2u\big)\hat{m}^{2/3}\bigr)_x=0, \qquad&& -\infty<x<\infty,\quad t>0,& \\
&\hat{m}\equiv m+1=u-u_{xx}+1, &&& \label{hat-m}\\
&u(x,0)=u_0(x), &&-\infty<x<\infty, & \label{IC1}
\end{alignat}
where $u_0(x)$ is suf\/f\/iciently smooth and decays fast as $x\to\pm\infty$. Moreover, we assume that $u_0(x)$ satisf\/ies the sign condition
\begin{gather} \label{IC-sign}
u_0(x)-(u_0)_{xx}(x)+1>0.
\end{gather}
\end{subequations}
Then there exists~\cite{LLW13} a unique global solution $u(x,t)$ of~\eqref{N-Cauchy}, such that $u(x,t)\to 0$ as $x\to\pm\infty$ for all~$t$.

Notice that the Novikov equation \eqref{N5}, being written in terms of $u$ only, contains linear as well as quadratic dispersion terms:
\begin{gather*}
u_t-u_{txx}+4u^2u_x+8uu_x+3u_x=3u u_xu_{xx}+u^2u_{xxx}+3u_xu_{xx}+2uu_{xxx}.
\end{gather*}

{\bf 1.6.}~The analysis of Camassa--Holm-type equations by using the inverse scattering approach was initiated in \cite{C01,CL03,J03,L02} for the Camassa--Holm equation itself:
\begin{gather*}
u_t-u_{txx}+3uu_x=2u_xu_{xx}+uu_{xxx}.
\end{gather*}
A version of the inverse scattering method for the CH equation based on a Riemann--Hilbert (RH) factorization problem was proposed in~\cite{BS06,BS08} (another RH formulation of the inverse scattering transform is presented in~\cite{CGI06}). The RH approach has proved its ef\/f\/iciency in the study of the long-time behavior of solutions of both initial value problems~\cite{BIS10,BKST,BS-D} and initial boundary value problems \cite{BS-F} for the CH equation. In~\cite{BLS15,BS13} it has been adapted to the study of the Degasperis--Procesi equation.

In the present paper we develop the RH approach to the Novikov equation in the form \eqref{N5} on zero background, following the main ideas developed in~\cite{BS08,BS13}. To the best of our knowledge, no equations of the Camassa--Holm type with \emph{cubic} nonlinearity have been treated before by the inverse scattering method in the form of a RH problem.

A major dif\/ference between the implementations of the RH method to the CH equation, on one hand, and to the DP as well as Novikov equation, on the other hand, is that in the latter cases, the spatial equations of the associated Lax pairs are of the third order, which implies that when rewriting them in matrix form, one has to deal with $3\times 3$ matrix-valued equations, while in the case of the CH equation, they have a $2\times 2$ matrix structure, as in the cases of the most known integrable equations (KdV, mKdV, nonlinear Schr\"odinger, sine-Gordon, etc.). Hence, the construction and analysis of the associated RH problem become considerably more complicated.

In our approach, we propose (Section~\ref{sec:3}) an associated RH problem and give (Theo\-rem~\ref{thm:main}) a~representation of the solution~$u(x,t)$ of the initial value problem~\eqref{N-Cauchy} in terms of the solution of this RH problem evaluated at a distinguished point of the plane of the spectral parameter. Remarkably, the formulas for $u(x,t)$ obtained in this way have the same structure as the parametric formulas obtained in~\cite{M13} for pure multisoliton solutions.

\section{Lax pairs and eigenfunctions}

\begin{ass*}
Recall that we assume that the initial function $u_0(x)$ in~\eqref{IC1} is suf\/f\/iciently smooth with fast decay at $\pm\infty$ and satisf\/ies the sign condition~\eqref{IC-sign}.
\end{ass*}

Then, similarly to the case of the CH equation (see, e.g.,~\cite{C01}), the solution $\hat m(x,t)$ of~\eqref{N-Cauchy} satisf\/ies the sign condition $\hat m(x,t)>0$ for all $x\in\D{R}$ and all $t>0$.

\subsection{Lax pairs} \label{lax.pairs}

\subsubsection{A f\/irst Lax pair}

The Lax pair found by Hone and Wang \cite{HW08} for the Novikov equation in the form \eqref{N1} (or \eqref{N2-0}) reads
\begin{subequations} \label{Hone-Wang}\allowdisplaybreaks
\begin{gather}
\partial_x\begin{pmatrix}
\psi_1 \\
\psi_2 \\
\psi_3
\end{pmatrix} = \begin{pmatrix}
0 & zm & 1 \\
0 & 0 & zm \\
1 & 0 & 0
\end{pmatrix}\begin{pmatrix}
\psi_1 \\
\psi_2 \\
\psi_3
\end{pmatrix},\qquad m\coloneqq u-u_{xx},\\
\partial_t\begin{pmatrix}
\psi_1 \\
\psi_2 \\
\psi_3
\end{pmatrix} = \begin{pmatrix}
-uu_x & \frac{u_x}{z}-u^2mz & u_x^2 \vspace{1mm}\\
\frac{u}{z} & -\frac{1}{z^2} & -\frac{u_x}{z}-u^2mz \vspace{1mm}\\
-u^2 & \frac{u}{z} & uu_x
\end{pmatrix}\begin{pmatrix}
\psi_1 \\
\psi_2 \\
\psi_3
\end{pmatrix},
\end{gather}
\end{subequations}
where $z$ is the \emph{spectral parameter}.

\subsubsection{A modif\/ied Lax pair}

For the Novikov equation in the form \eqref{N5}, the Lax pair \eqref{Hone-Wang} has to be appropriately modif\/ied. While for the Camassa--Holm and Degasperis--Procesi equations the corresponding modif\/ication (when passing from the equation with zero linear dispersive term to that with a non-zero one) consists simply in replacing $m$ by $\hat m=m+1$, the modif\/ication for the Novikov equation turns out to be more involved.

\begin{Lemma}\label{lax-mod}
The Novikov equation~\eqref{N5} admits as Lax pair the system
\begin{subequations} \label{Lax-vec}
\begin{gather} 
\Phi_x = U \Phi,\qquad
\Phi_t = V \Phi,
\end{gather}
where $\Phi\equiv\Phi(x,t;z)$ and
\begin{gather}\label{U}
U(x,t;z)=\begin{pmatrix}
0 & z\hat m & 1 \\
0 & 0 & z\hat m \\
1 & 0 & 0
\end{pmatrix},\\
\label{V}
V(x,t;z)=\begin{pmatrix}
-(u+1)u_x+\frac{1}{3z^2}&\frac{u_x}{z}-(u^2+2u)\hat mz&u_x^2+1\\[1mm]
\frac{u+1}{z}&-\frac{2}{3z^2}&-\frac{u_x}{z}-(u^2+2u)\hat mz\\[1mm]
-u^2-2u&\frac{u+1}{z}&(u+1)u_x+\frac{1}{3z^2}
\end{pmatrix}.
\end{gather}
\end{subequations}
\end{Lemma}

\begin{Remark}
The freedom in adding to $V$ a constant (independent of $(x,t)$) term $c\cdot I$, where $I$ is the $3\times 3$ identity matrix, has been used in \eqref{V} in order to make $V$ traceless, which provides that the determinant of a matrix solution to the equation $\Phi_t=V\Phi$ is independent of~$t$. The same property holds for the equation $\Phi_x=U\Phi$ whose coef\/f\/icient $U$ is obviously traceless.
\end{Remark}

The coef\/f\/icient matrices $U$ and $V$ in~\eqref{Lax-vec} have singularities (in the extended complex $z$-plane) at $z=0$ and at $z=\infty$. In order to control the behavior of solutions to~\eqref{Lax-vec} as functions of the spectral parameter $z$ (which is crucial for the Riemann--Hilbert method), we follow a strategy similar to that adopted for the CH equation \cite{BS06,BS08}
and the DP equation~\cite{BS13}.

\subsubsection[A Lax pair appropriate for large $z$]{A Lax pair appropriate for large $\boldsymbol{z}$}

In order to control the large $z$ behavior of the solutions of \eqref{Lax-vec}, we will transform this Lax pair as follows (cf.~\cite{BC}):

\begin{Lemma}\label{lax-Q}
The Lax pair \eqref{Lax-vec} can be transformed into a new Lax pair
\begin{gather}\label{Lax-Q-form}
\hat\Phi_x-Q_x\hat\Phi=\hat U\hat\Phi,\qquad
\hat\Phi_t-Q_t\hat\Phi=\hat V\hat\Phi,
\end{gather}
whose coefficients $Q(x,t;z)$, $\hat U(x,t;z)$, and $\hat V(x,t;z)$ have the following properties:
\begin{enumerate}\itemsep=0pt
\item[\rm(i)] $Q$ is diagonal;
\item[\rm(ii)] $\hat U=\ord(1)$ and $\hat V=\ord(1)$ as $z\to\infty$, whereas $Q_x$ is growing;
\item[\rm(iii)] the diagonal parts of $\hat U$ and $\hat V$ decay as $z\to\infty$;
\item[\rm(iv)] $\hat U\to 0$ and $\hat V\to 0$ as $x\to\pm\infty$.
\end{enumerate}
\end{Lemma}

\begin{proof}
As in the case of the DP equation~\cite{BS13}, we perform this transformation into two steps:{\samepage
\begin{itemize}\itemsep=0pt
\item We transform \eqref{Lax-vec} into a system where the leading terms are represented as products of $(x,t)$-independent (matrix-valued) and $(x,t)$-dependent (scalar) factors.
\item We diagonalize the $(x,t)$-independent factors.
\end{itemize}}

{\bf First step.} Introducing $\tilde\Phi\equiv\tilde\Phi(x,t;z)$ by
\begin{gather*}
\tilde\Phi=D^{-1}\Phi,
\end{gather*}
where
$D(x,t)=\diag\big\{q(x,t),1,q^{-1}(x,t)\big\}$
and ($\hat m$ is as in \eqref{hat-m})
\begin{gather} \label{q}
q=q(x,t)\coloneqq\hat m^{1/3}(x,t),
\end{gather}
transforms \eqref{Lax-vec} into the new Lax pair
\begin{subequations} \label{Lax-vec-1}
\begin{gather} 
\tilde\Phi_x = \tilde U \tilde\Phi,\qquad
\tilde\Phi_t = \tilde V \tilde\Phi,
\end{gather}
where
\begin{gather}
\tilde U(x,t;z)=q^2(x,t)
\begin{pmatrix}
0 & z & 1 \\
0 & 0 & z \\
1 & 0 & 0
\end{pmatrix}+\begin{pmatrix}
-\frac{q_x}{q} & 0 & \frac{1}{q^2}-q^2 \\
0 & 0 & 0 \\
0 & 0 & \frac{q_x}{q}
\end{pmatrix}\nonumber\\
\label{U-tilde}
\hphantom{\tilde U(x,t;z)}{}\equiv q^2(x,t) U_\infty (z) + \tilde U^{(1)}(x,t)\\
\tilde V(x,t;z)= - \big(u^2+2u\big) q^2 \begin{pmatrix}
0 & z & 1 \\
0 & 0 & z \\
1 & 0 & 0
\end{pmatrix} + \begin{pmatrix}
\frac{1}{3z^2} & 0 & 1\vspace{1mm} \\
\frac{1}{z} & -\frac{2}{3 z^2} & 0\vspace{1mm} \\
0 & \frac{1}{z} & \frac{1}{3 z^2}
\end{pmatrix}\nonumber\\
\hphantom{\tilde V(x,t;z)=}{} + \begin{pmatrix}
(u^2+2u)\frac{q_x}{q} & 0 & (u^2+2u)q^2 +\frac{u_x^2+1}{q^2}-1\\
0 & 0 & 0 \\
0 & 0 & -(u^2+2u)\frac{q_x}{q}
\end{pmatrix}\nonumber\\
\hphantom{\tilde V(x,t;z)=}{}
+\frac{1}{z}
\begin{pmatrix}
0 & \frac{u_x}{q} & 0 \\
(u+1)q-1 & 0 & -\frac{u_x}{q} \\
0 & (u+1)q-1 & 0
\end{pmatrix}\nonumber\\
\label{V-tilde}
\hphantom{\tilde V(x,t;z)}{}
\equiv -\big(u^2(x,t)+2u(x,t)\big)q^2(x,t) U_\infty (z) + V_\infty (z) + \tilde V^{(1)}(x,t) + \frac{1}{z}\tilde V^{(2)}(x,t).
\end{gather}
\end{subequations}

{\bf Second step.} The commutator of $U_\infty$ and $V_\infty$ vanishes identically, i.e., $[U_\infty, V_\infty]\equiv 0$, which allows simultaneous diagonalization of $U_\infty$ and $V_\infty$. Indeed, we have (for $z\neq 0$)
\begin{gather*}
P^{-1}(z)U_\infty (z)P(z)=\Lambda(z),\qquad
P^{-1}(z)V_\infty (z)P(z)=A(z),
\end{gather*}
where
\begin{subequations}\allowdisplaybreaks \label{La-P}
\begin{gather}\label{La}
\Lambda(z) = \begin{pmatrix}
\lambda_1(z) & 0 & 0 \\
0 & \lambda_2(z) & 0 \\
0 & 0 & \lambda_3(z)
\end{pmatrix},\\
\label{Aa}
A(z)=\frac{1}{3z^2}I+\Lambda^{-1}(z)
\equiv
\begin{pmatrix}
A_1(z) & 0 & 0 \\
0 &A_2(z) & 0 \\
0 & 0 &A_3(z)
\end{pmatrix},\\
\label{P}
P(z) = \begin{pmatrix}
\lambda_1^2(z) & \lambda_2^2(z) & \lambda_3^2(z) \\
z & z & z \\
\lambda_1(z) & \lambda_2(z) & \lambda_3(z)
\end{pmatrix},\\
\label{La-P-1}
P^{-1}(z) = \left(\begin{matrix}
(3\lambda_1^2(z)-1)^{-1}\! & 0 & 0 \\
0 & (3\lambda_2^2(z)-1)^{-1} & 0 \\
0 & 0 & \!(3\lambda_3^2(z)-1)^{-1}
\end{matrix}\right)\!
\left(\begin{matrix}
1 & \frac{z}{\lambda_1(z)} & \lambda_1(z) \\
1 & \frac{z}{\lambda_2(z)} & \lambda_2(z)\\
1 & \frac{z}{\lambda_3(z)} & \lambda_3(z)
\end{matrix}\right).\!\!\!
\end{gather}
\end{subequations}
Here $\lambda_1(z)$, $\lambda_2(z)$, and $\lambda_3(z)$ are the solutions of the algebraic equation
\begin{gather*}
\lambda^3 - \lambda=z^2,
\end{gather*}
enumerated in such a way that $\lambda_j(z)\sim\omega^jz^{2/3}$ as $z\to\infty$, where $\omega=\eul^{\frac{2\ii\pi}{3}}$.

Now, introducing $\hat\Phi\equiv\hat\Phi(x,t;z)$ by
\begin{gather*}
\hat\Phi=P^{-1}\tilde\Phi
\end{gather*}
transforms the Lax pair \eqref{Lax-vec-1} into the new Lax pair:
\begin{subequations} \label{Lax-vec-2}
\begin{gather}
\hat\Phi_x - q^2\Lambda(z)\hat\Phi = \hat U \hat\Phi,\\
\hat\Phi_t + \left(\big(u^2+2u\big)q^2\Lambda(z) - A(z)\right)\hat\Phi = \hat V \hat\Phi,
\end{gather}
where
\begin{gather} \label{hat-U}
\hat U(x,t;z)= P^{-1}(z) \tilde U^{(1)}(x,t) P(z),\\
\hat V(x,t;z)= P^{-1}(z)\left(\tilde V^{(1)}(x,t)+\frac{1}{z}\tilde V^{(2)}(x,t)\right)P(z).
\end{gather}
\end{subequations}
Here $\hat U(x,t;z)=\ord(1)$ and $\hat V(x,t;z)=\ord(1)$ as $z\to\infty$ due to the fact that $\tilde U^{(1)}$ and $\tilde V^{(1)}$ are upper triangular matrices. Moreover, the fact that $\tilde U^{(1)}$ and $\tilde V^{(1)}$ are traceless implies that the diagonal entries of $\hat U(x,t;z)$ and $\hat V(x,t;z)$ are $\ord(z^{-2/3})$ as $z\to\infty$.

Indeed, we can write $\hat U=\hat U^{(1)}\hat U^{(2)}$, where
\begin{subequations} \label{hatUV12}
\begin{gather}
\hat U^{(1)}=\left(\begin{matrix}
\frac{1}{3\lambda_1^2-1} & 0 & 0 \\
0 & \frac{1}{3\lambda_2^2-1} & 0 \\
0 & 0 & \frac{1}{3\lambda_3^2-1}
\end{matrix}\right),\\
\hat U^{(2)}=
\left(\begin{matrix}
c_2\lambda_1
 & c_1(\lambda_1\lambda_2-\lambda_2^2)+ c_2\lambda_2
	& c_1(\lambda_1\lambda_3-\lambda_3^2)+c_2\lambda_3 \vspace{1mm}\\
	c_1(\lambda_1\lambda_2-\lambda_1^2)+ c_2\lambda_1
	& c_2\lambda_2
	& c_1(\lambda_1\lambda_3-\lambda_3^2)+ c_2\lambda_3 \vspace{1mm}\\
	c_1(\lambda_1\lambda_3-\lambda_1^2)+ c_2\lambda_1
	& c_1(\lambda_2\lambda_3-\lambda_2^2)+ c_2\lambda_2
	& c_2\lambda_3
\end{matrix}\right)
\end{gather}
with $c_1 = q_x/q$ and $c_2=q^{-2}-q^2$. Notice that $\hat U(x,t;z)$ has a f\/inite limit at $z=0$.

We can also write $\hat V=\hat U^{(1)}\big(\hat V^{(1)}+\hat V^{(2)}\Lambda\big)$ where $\hat V^{(1)}$ has the form of $\hat U^{(2)}$ with $c_1$ and $c_2$ replaced by $c_3=-(u^2+2u)\frac{q_x}{q}$ and $c_4=(u^2+2u)q^2+\frac{u_x^2+1}{q^2}-1$, respectively, and{\samepage
\begin{gather}
\hat V^{(2)} =
\left(\begin{matrix}
2c_6
 & \begin{matrix} c_5\left(\frac{1}{\lambda_2}-\frac{1}{\lambda_1}\right)\\
 {}+c_6\left(\frac{\lambda_1}{\lambda_2}+\frac{\lambda_2}{\lambda_1}\right)\end{matrix}
	& \begin{matrix} c_5\left(\frac{1}{\lambda_3}-\frac{1}{\lambda_1}\right)\\
{} + c_6\left(\frac{\lambda_1}{\lambda_3}+\frac{\lambda_3}{\lambda_1}\right)\end{matrix} \vspace{1mm}\\
\begin{matrix}	c_5\left(\frac{1}{\lambda_1}-\frac{1}{\lambda_2}\right)\\
{} + c_6\left(\frac{\lambda_1}{\lambda_2}+ \frac{\lambda_2}{\lambda_1}\right)\end{matrix}
	& 2c_6
	& \begin{matrix} c_5\left(\frac{1}{\lambda_3}-\frac{1}{\lambda_2}\right)\\
{} + c_6\left(\frac{\lambda_2}{\lambda_3}+ \frac{\lambda_3}{\lambda_2}\right)\end{matrix} \vspace{1mm}\\
\begin{matrix}	c_5\left(\frac{1}{\lambda_1}-\frac{1}{\lambda_3}\right)\\
{} + c_6\left(\frac{\lambda_1}{\lambda_3}+ \frac{\lambda_3}{\lambda_1}\right)\end{matrix}
& \begin{matrix} c_5\left(\frac{1}{\lambda_2}-\frac{1}{\lambda_3}\right)\\
{} + c_6\left(\frac{\lambda_2}{\lambda_3}+ \frac{\lambda_3}{\lambda_2}\right)\end{matrix}
	& 2 c_6
\end{matrix}\right)
\end{gather}
with $c_5=\frac{u_x}{q}$, $c_6=(u+1)q-1$.}
\end{subequations}

Finally, in order to write \eqref{Lax-vec-2} in the desired form \eqref{Lax-Q-form} it suf\/f\/ices to f\/ind a solution of the system
\begin{subequations} \label{Q-dif}
\begin{gather} \label{Q-dif-x}
Q_x=q^2\Lambda(z),\\
\label{Q-dif-t}
Q_t=-\left(\big(u^2+2u\big)q^2\right)\Lambda(z) + A(z).
\end{gather}
\end{subequations}
We f\/irst notice that both equations are consistent. This follows directly from the equation
\begin{gather}\label{N-q}
(q^2)_t+\big(\big(u^2+2u\big)q^2\big)_x=0,
\end{gather}
which is just another form of the Novikov equation \eqref{N5}. We actually f\/ind that a solution $Q(x,t;z)$ of the system~\eqref{Q-dif} is given by the $3\times 3$ diagonal function
\begin{gather}\label{Q}
Q(x,t;z)=y(x,t)\Lambda(z)+tA(z)
\end{gather}
with
\begin{gather}\label{y}
y(x,t)\coloneqq x-\int_x^{\infty}(q^2(\xi,t)-1)\dd\xi.
\end{gather}
This solution $Q(x,t;z)$ is normalized in such a way that
\begin{gather*}
Q(x,t;z)\sim x\Lambda(z)+tA(z)\qquad\text{as} \quad x\to+\infty.\tag*{\qed}
\end{gather*}
\renewcommand{\qed}{}
\end{proof}

\subsection{Eigenfunctions} \label{eigenfunctions}

\subsubsection{Fredholm integral equations}

Introducing the $3\times 3$ matrix-valued function $M\equiv M(x,t;z)$ by
\begin{gather*}
M=\hat\Phi\eul^{-Q}
\end{gather*}
reduces the Lax pair \eqref{Lax-Q-form} to the system
\begin{gather}\label{M} 
M_x-[Q_x,M]=\hat UM,\qquad
M_t-[Q_t,M]=\hat VM.
\end{gather}
Assume that the coef\/f\/icients in~\eqref{M}, which are expressed in terms of $u(x,t)$, are given. Then particular solutions of~\eqref{M} having well-controlled properties as functions of the spectral parameter $z$ can be constructed as solutions of the Fredholm integral equation (cf.~\cite{BC})
\begin{gather}
M(x,t;z)=I \nonumber\\
\label{M-int}
{}+ \int_{(x^*,t^*)}^{(x,t)}\eul^{Q(x,t;z)-Q(\xi,\tau;z)}\big( \hat U M(\xi,\tau;z)\dd\xi + \hat V M(\xi,\tau;z)\dd\tau\big)\eul^{-Q(x,t;z)+Q(\xi,\tau;z)},
\end{gather}
where the initial points of integration $(x^*,t^*)$ can be chosen dif\/ferently for dif\/ferent matrix entries of the equation. $Q$ being diagonal, \eqref{M-int} must be seen as the collection of scalar integral equations ($1\leq j,l\leq 3$)
\begin{gather*}
M_{jl}(x,t;z)=I_{jl}\\
{} +\int_{(x_{jl}^*,t_{jl}^*)}^{(x,t)}\eul^{Q_{jj}(x,t;z)-Q_{jj}(\xi,\tau;z)}\bigl((\hat UM)_{jl}(\xi,\tau;z)\dd\xi+(\hat VM)_{jl}(\xi,\tau;z)\dd\tau\bigr)\eul^{-Q_{ll}(x,t;z)+Q_{ll}(\xi,\tau;z)}.
\end{gather*}
Notice that choosing the $(x_{jl}^*,t_{jl}^*)$ appropriately allows obtaining eigenfunctions which are piecewise analytic w.r.t.\ the spectral parameter $z$ and thus can be used in the construction of Riemann--Hilbert problems associated with initial value problems~\cite{BS08} as well as initial boundary value problems~\cite{BS-F}.

\begin{figure}[ht]
\centering\includegraphics[scale=1]{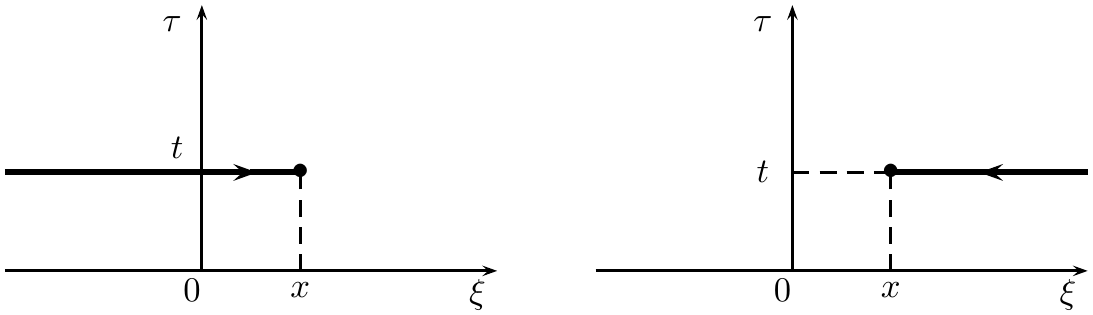}
\caption{Paths of integration. Left: $(x^*,t^*)=(-\infty,t)$. Right: $(x^*,t^*)=(+\infty,t)$.}
\label{fig:paths}
\end{figure}

In particular, for the Cauchy problem considered in the present paper, it is reasonable to choose these points to be $(-\infty,0)$ or $(+\infty,0)$ thus reducing the integration in \eqref{M-int} to paths parallel to the $x$-axis (see Fig.~\ref{fig:paths}) provided the integrals $\int_{(-\infty,0)}^{(-\infty,t)}$ and $\int_{(\infty,0)}^{(\infty,t)}$ vanish:
\begin{gather} \label{M-int1}
M(x,t;z)=I+\int_{(\pm)\infty}^x\eul^{Q(x,t;z)-Q(\xi,t;z)}\croch{\hat UM(\xi,t;z)}\eul^{-Q(x,t;z)+Q(\xi,t;z)}\dd\xi,
\end{gather}
or, in view of \eqref{Q} and \eqref{Q-dif-x},
\begin{gather} \label{M-int2}
M(x,t;z)=I+\int_{(\pm)\infty}^x\eul^{-\left(\int_x^{\xi}q^2(\zeta,t)\dd\zeta\right)\Lambda(z)}\croch{\hat UM(\xi,t;z)}\eul^{\big(\int_x^{\xi}q^2(\zeta,t)\dd\zeta\big)\Lambda(z)}\dd\xi.
\end{gather}
Since $q^2>0$, the domains (in the complex $z$-plane), where the exponential factors in~\eqref{M-int1} are bounded, are determined by the signs of $\Re\lambda_j(z)-\Re\lambda_l(z)$, $1\leq j\neq l\leq 3$.

\subsubsection{A new spectral parameter}

As in the case of the Degasperis--Procesi equation (see \cite{BS13, CIL10} and also~\cite{Ca82}) it is convenient to introduce a new spectral parameter $k$ such that
\begin{gather} \label{z-k}
z^2(k)=\frac{1}{3\sqrt{3}}\left(k^3+\frac{1}{k^3}\right).
\end{gather}
We thus have
\begin{gather} \label{la-k}
\lambda_j=\lambda_j(z(k))=\frac{1}{\sqrt{3}}\left(\omega^jk+\frac{1}{\omega^jk}\right),\qquad \text{where}\quad \omega=\eul^{\frac{2\ii\pi}{3}}.
\end{gather}
In what follows, we will work in the complex $k$-plane only. So, by a slight abuse of notation, we will write $\lambda_j(k)$ for $\lambda_j(z(k))$, and similarly for other functions of $z(k)$, e.g., $M(x,t;k)$ for $M(x,t;z(k))$ and $\Lambda(k)$ for $\Lambda(z(k))$.

\begin{figure}[ht]
\centering\includegraphics[scale=.75]{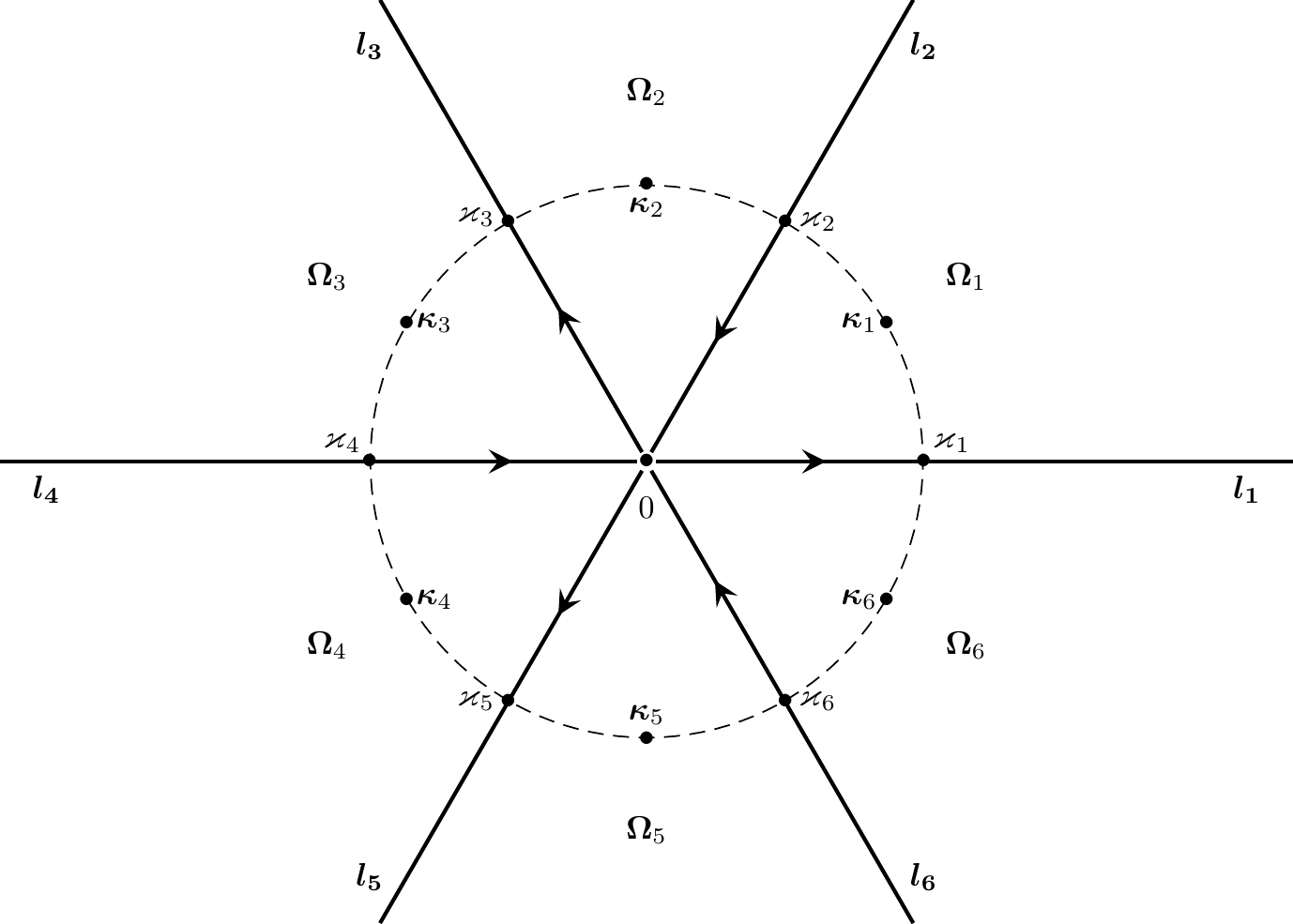}
\caption{Rays $l_{\nu}$, domains $\Omega_{\nu}$, and points $\varkappa_{\nu}$, $\kappa_l$ in the $k$-plane.}
\label{fig:sectors}
\end{figure}

The $\lambda_j$'s are the same, as functions of $k$, as in the case of the DP equation, see~\cite{BS13}. Thus, the contour $\Sigma=\{k\,|\,\Re\lambda_j(k)=\Re\lambda_l(k)\text{ for some } j\neq l\}$ is also the same; it consists of six rays
\begin{gather*}
l_{\nu}=\D{R}_+\eul^{\frac{\ii\pi}{3}(\nu-1)}=\D{R}_+\varkappa_{\nu},\qquad\nu=1,\dots,6
\end{gather*}
dividing the $k$-plane into six sectors
\begin{gather*}
\Omega_{\nu}=\left\{ k \,\Big\vert\, \frac{\pi}{3}(\nu-1)<\arg k < \frac{\pi}{3}\nu\right\},\qquad\nu=1,\dots,6.
\end{gather*}

In order that \eqref{M-int2} have a (matrix-valued) solution that is analytic, as a function of $k\in\D{C}\setminus\Sigma$, the initial points of integration $\infty_{jl}$ are specif\/ied as follows for each matrix entry $(j,l)$, $1\leq j,l\leq 3$:
\begin{gather}\label{signs}
\infty_{j,l}=
\begin{cases}
+\infty,&\text{if }\Re\lambda_j(k)\geq\Re\lambda_l(k),\\
-\infty, &\text{if }\Re\lambda_j(k)<\Re\lambda_l(k).
\end{cases}
\end{gather}
This means that we consider the system of scalar Fredholm integral equations, $1\leq j,\,l\leq 3$,
\begin{gather} \label{M-int3}
M_{jl}(x,t;k)=I_{jl}+\int_{\infty_{j,l}}^x\eul^{-\lambda_j(k)\int_x^{\xi}q^2(\zeta,t)\dd\zeta}\croch{(\hat UM)_{jl}(\xi,t;k)}\eul^{\lambda_l(k)\int_x^{\xi}q^2(\zeta,t)\dd\zeta}\dd\xi,
\end{gather}
where $k\in\D{C}$ and $I$ denotes the $3\times 3$ identity matrix.

\begin{Remark}
In spite of the fact that some of the coef\/f\/icients in \eqref{Lax-vec-2} seemingly depend on the f\/irst order of $z$ (for instance, the coef\/f\/icient $\frac{1}{z}P^{-1}(z)\tilde V^{(2)}P(z)$), which, as a function of $k$, is not rational, direct calculations show that $\hat U$ in~\eqref{M-int3} as well as $\hat V$ depends on $k$ rationally, through the $\lambda_j=\lambda_j(k)$'s, see~\eqref{hatUV12} and~\eqref{la-k}.
\end{Remark}

\begin{Proposition}[analyticity] \label{prop:p1}
Let $M(x,t;k)$ be the $($unique$)$ solution of the system of integral equations~\eqref{M-int3}, where the limits of integration $\infty_{j,l}$ are chosen according to~\eqref{signs}. Then
\begin{enumerate}\itemsep=0pt
\item[\rm(i)] $M$ is piecewise meromorphic with respect to $\Sigma$, as function of the spectral parameter $k$;
\item[\rm(ii)] $M(x,t;k)\to I$ as $k\to\infty$, where $I$ is the $3\times 3$ identity matrix;
\item[\rm(iii)] for $k\in\D{C}\setminus\Sigma$, $M$ is bounded as $x\to-\infty$ and $M\to I$ as $x\to+\infty$;
\item[\rm(iv)] $\det M\equiv 1$.
\end{enumerate}
\end{Proposition}

\begin{proof}
The proof follows the same lines as in~\cite{BC}. Notice that in order to have (ii), it is important that the diagonal part of $\hat U$ vanish as $k\to\infty$. Notice also that (iv) follows from the fact that the coef\/f\/icient matrices in~\eqref{Lax-vec-2} are traceless and from (ii).
\end{proof}

\begin{Proposition}[symmetries] \label{p2}
The solution $M(x,t;k)$ of \eqref{M-int3} satisfies the symmetry relations:
\begin{enumerate}\itemsep=0pt
\item[\rm{(S}1)]
$\Gamma_1\overline{M(x,t;\bar k)} \Gamma_1 = M(x,t;k)$, where $\Gamma_1 =
\left(\begin{smallmatrix}
0 & 1 & 0 \\
1 & 0 & 0 \\
0 & 0 & 1
\end{smallmatrix}\right)$;
\item[\rm{(S}2)]
$\Gamma_2 \overline{M(x,t;\bar k\omega^2)} \Gamma_2 = M(x,t;k)$, where $\Gamma_2 =
\left(\begin{smallmatrix}
	0 & 0 & 1 \\
	0 & 1 & 0 \\
	1 & 0 & 0
\end{smallmatrix}\right)$;
\item[\rm{(S}3)]
$\Gamma_3 \overline{M(x,t;\bar k\omega)} \Gamma_3 = M(x,t;k)$, where $\Gamma_3 =
\left(\begin{smallmatrix}
	1 & 0 & 0 \\
	0 & 0 & 1 \\
	0 & 1 & 0
\end{smallmatrix}\right)$;
\item[\rm{(S}4)]
$M(x,t;\frac{1}{k})=\overline{M(x,t;\bar k)}$.
\end{enumerate}
\end{Proposition}

\begin{proof} Indeed, the diagonal entries of the matrix $\Lambda(k)=\diag(\lambda_j(k))$ satisfy the following relations:
\begin{alignat*}{4}
&\overline{\lambda_1(\bar k)} = \lambda_2(k), \quad&& \overline{\lambda_3(\bar k)} =
\lambda_3(k), &&&\\
&\overline{\lambda_2(\bar k \omega)} = \lambda_3(k),\qquad && \overline{\lambda_1(\bar k \omega)}=\lambda_1(k), &&&\\
&\overline{\lambda_1(\bar k \omega^2)} = \lambda_3(k),\qquad &&\overline{\lambda_2(\bar k \omega^2)}=\lambda_2(k), \qquad &&
&\lambda_j(1/k) = \overline{\lambda_j(\bar k)},\quad j=1,2,3.&\tag*{\qed}
\end{alignat*}
\renewcommand{\qed}{}
\end{proof}

\begin{Remark}
From (S1)--(S3) it follows that the values of $M$ at $k$ and at $\omega k$ are related by
\begin{gather*}
M(x,t;k\omega) = C^{-1}M(x,t;k)C,\qquad \text{where}\quad
C=\left(\begin{matrix}
	0 & 0 & 1 \\
	1 & 0 & 0 \\
	0 & 1 & 0
\end{matrix}\right).
\end{gather*}
\end{Remark}

If $\lambda_j(k)=\lambda_l(k)$, $j\neq l$ for some value of the spectral parameter~$k$, then~$P$ at this value becomes degenerate (see~\eqref{La-P-1}), which in turn leads to a singularity for $\hat U$ and, consequently, for~$\hat\Phi$ and~$M$. These particular values of the spectral parameter are $\varkappa_{\nu}=\eul^{\frac{\ii\pi}{3}(\nu-1)}$, $\nu=1,\dots,6$. Taking into account the symmetries described in Proposition~\ref{p2} leads to the following proposition.

\begin{Proposition}[singularities at $\varkappa_{\nu}$] \label{p3}
The limiting values of $M(x,t;k)$ as $k$ approaches one of the points $\varkappa_{\nu}=\eul^{\frac{\ii\pi}{3}(\nu-1)}$, $\nu=1,\dots,6$ have pole singularities with leading terms of a specific matrix structure.
\begin{enumerate}\itemsep=0pt
\item[\rm(i)] As $k\to\varkappa_1=1$,
\begin{subequations} \label{singul}
\begin{gather}
M(x,t;k)=\frac{1}{k-1}
\left(\begin{matrix}
	1 & 1 & 1 \\
	-1 & -1 & -1 \\
	0 & 0 & 0
\end{matrix}\right)
\left(\begin{matrix}
	\alpha(x,t) & 0 & 0 \\
	0 & \alpha(x,t) & 0 \\
	0 & 0 & \beta(x,t)
\end{matrix}\right)\nonumber\\
\hphantom{M(x,t;k)=}{} +\tilde M+\ord(k-1),\label{singul-1}\\
M^{-1}(x,t;k)=\frac{1}{k-1}
\left(\begin{matrix}
	1 & 1 & 1 \\
	-1 & -1 & -1 \\
	0 & 0 & 0
\end{matrix}\right)
\left(\begin{matrix}
	\alpha_1(x,t) & 0 & 0 \\
	0 & \alpha_1(x,t) & 0 \\
	0 & 0 & \beta_1(x,t)
\end{matrix}\right)\nonumber\\
\hphantom{M^{-1}(x,t;k)=}{} +\tilde M^{(1)}+\ord(k-1),\label{singul-1-1}
\end{gather}
\end{subequations}
where $\alpha=-\overline{\alpha}$, $\beta=-\overline{\beta}$, $\alpha_1=-\overline{\alpha_1}$, $\beta_1=-\overline{\beta_1}$. Moreover, $(\alpha,\beta)\neq(0,0)$ iff $(\alpha_1,\beta_1)\neq(0,0)$ and in this case, the entries of $\tilde M(x,t)$ and $\tilde M^{(1)}(x,t)$ satisfy the relations
\begin{subequations}\label{M-prop}
\begin{alignat}{3}\label{M-prop-a}
&\tilde M_{31}=\tilde M_{32},\qquad&&\tilde M_{11}+\tilde M_{21}=\tilde M_{12}+\tilde M_{22}, & \\
\label{M-prop-b}
&\tilde M_{31}^{(1)}=\tilde M_{32}^{(1)},\qquad &&\tilde M_{11}^{(1)}+\tilde M_{21}^{(1)}=\tilde M_{12}^{(1)}+\tilde M_{22}^{(1)}.&
\end{alignat}
\end{subequations}
\item[\rm(ii)]
As $k\to\varkappa_2=\eul^{\frac{\ii\pi}{3}}$, properties \eqref{singul} hold with $k-1$ replaced by $k-\eul^{\frac{\ii\pi}{3}}$. Relations~\eqref{M-prop} also hold provided we replace the indices $1$, $2$, $3$ of the matrix entries by $2$, $3$, $1$, respectively.
\item[\rm(iii)] $M$ and $M^{-1}$ have similar leading terms at the other polar singularities $\varkappa_3,\dots,\varkappa_6$ in accordance with the symmetry conditions stated in Proposition~{\rm \ref{p2}}.
\end{enumerate}
\end{Proposition}

\begin{proof}
(i-1) Let $\check M\coloneqq P(k)M(x,t;k)P^{-1}(k)$. By \eqref{M-int3} this function satisf\/ies the integral equation
\begin{gather*}
\check M(x,t;k)=I \\
{}+\int_{(\pm)\infty}^xP(k)\eul^{-\Lambda(k)\int_x^{\xi}q^2(\zeta,t)\dd\zeta}
P^{-1}(k)\croch{(\tilde U^{(1)}\check M)(\xi,t;k)}P(k)\eul^{\Lambda(k)\int_x^{\xi}q^2(\zeta,t)\dd\zeta}P^{-1}(k)\dd\xi.
\end{gather*}

We f\/irst show that, in spite of the singularity of $P^{-1}(k)$ at $k=1$, $\check M$ is regular at this point. It suf\/f\/ices to show that $P(k)\eul^{\pm\Lambda(k)\int_x^{\xi}q^2(\zeta,t)\dd\zeta}P^{-1}(k)$ is non-singular at $k=1$. As $k\to 1$,
\begin{gather*}
\lambda_1(k)=-\frac{1}{\sqrt{3}}\big(1-\ii\sqrt{3}(k-1)\big)+\ord\big((k-1)^2\big)\qquad \text{and} \\ \lambda_2(k)=-\frac{1}{\sqrt{3}}\big(1+\ii\sqrt{3}(k-1)\big)+\ord\big((k-1)^2\big),
\end{gather*}
hence $(3\lambda_1^2(k)-1)^{-1}=c/(k-1)+\ord(1)$ and $(3\lambda_2^2(k)-1)^{-1}=-c/(k-1)+\ord(1)$ with $c=\frac{\ii}{2\sqrt{3}}$\,, whereas $\lambda_3(1)=2/\sqrt{3}$ and $(3\lambda_3^2(1)-1)^{-1}=1/3$. Thus, according to \eqref{La-P-1} we f\/ind that
\begin{gather*}
P^{-1}(k)=\frac{c}{k-1}\begin{pmatrix}
	a_1 & a_2 & a_3\\
	-a_1 & -a_2 & -a_3\\
	0 & 0 & 0
\end{pmatrix}+\ord(1)\qquad \text{as} \quad k\to 1
\end{gather*}
with $a_1=1$, $a_2=-\sqrt{2}/\sqrt[4]{3}$, and $a_3=-1/\sqrt{3}$. On the other hand, since $\lambda_1(1)=\lambda_2(1)$ the f\/irst two columns of $P(1)$ are equal (see \eqref{P}) and the f\/irst two diagonal entries of the diagonal matrix $\eul^{\pm\Lambda(1)\int_x^{\xi}q^2(\zeta,t)\dd\zeta}$ are the same. Then, the f\/irst two columns of the product $P(1)\eul^{\pm\Lambda(1)\int_x^{\xi}q^2(\zeta,t)\dd\zeta}$ are the same, and $P(k)\eul^{\pm\Lambda(k)\int_x^{\xi}q^2(\zeta,t)\dd\zeta}P^{-1}(k)$ is regular at $k=1$ since its polar part vanishes:
\begin{gather*}
\frac{c}{k-1}P(1)\eul^{\pm\Lambda(1)\int_x^{\xi}q^2(\zeta,t)\dd\zeta}
\begin{pmatrix}
	a_1 & a_2 & a_3\\
	-a_1 & -a_2 & -a_3\\
	0 & 0 & 0
\end{pmatrix}\equiv 0.
\end{gather*}
Thus, $\check M$ is also regular at $k=1$.

Since $M(x,t;k)=P^{-1}(k)\check M(x,t;k)P(k)$ where the last two factors are regular at $k=1$, the leading term of $M$ at $k=1$ is
\begin{gather*}
\frac{1}{k-1}
\begin{pmatrix}
	1 & 1 & 1 \\
	-1 & -1 & -1 \\
	0 & 0 & 0
\end{pmatrix}
\begin{pmatrix}
	ca_1 & 0 & 0 \\
	0 & ca_2 & 0 \\
	0 & 0 & ca_3
\end{pmatrix}\check M(x,t;1)P(1).
\end{gather*}
The f\/irst two columns of $P(1)$ being equal, it is the same for the product $R(x,t)$ of the last three factors. For multiplication by the f\/irst factor we can replace $R$ by the diagonal matrix whose diagonal entries are the sums $R_{1j}+R_{2j}+R_{3j}$. Thus we arrive at~\eqref{singul-1} with some~$\alpha(x,t)$ and~$\beta(x,t)$. The relations $\alpha=-\bar\alpha$ and $\beta=-\bar\beta$ come from the symmetry~(S1) stated in Proposition~\ref{p2}.

(i-2) Similarly, $\check M^{-1}=PM^{-1}P^{-1}$ also satisf\/ies an integral equation with non-singular coef\/f\/icients:{\samepage
\begin{gather*} 
\check M^{-1}(x,t;k)=I\\
{}-\int_{(\pm)\infty}^x P(k)\eul^{-\Lambda(k)\int_x^{\xi}q^2(\zeta,t)\dd\zeta}P^{-1}(k)\croch{\big(\check M^{-1}\tilde U^{(1)}\big)(\xi,t;k)}P(k)\eul^{\Lambda(k)\int_x^{\xi}q^2(\zeta,t)\dd\zeta}P^{-1}(k)\dd\xi,
\end{gather*}
which gives \eqref{singul-1-1}.}

(i-3) Using the fact that $\det M\equiv 1$, we calculate $M^{-1}$ starting from~\eqref{singul-1}, then, comparing the result with~\eqref{singul-1-1} we get dif\/ferent expressions of~$\alpha_1$ and~$\beta_1$ as linear combinations of~$\alpha$ and~$\beta$. Thus, $\alpha=\beta=0$ implies $\alpha_1=\beta_1=0$. Moreover, by comparison of these dif\/ferent expressions of~$\alpha_1$ and~$\beta_1$ we get the relations~\eqref{M-prop-a} for~$\tilde M$ provided~$\alpha$ or~$\beta$ is $\neq 0$. Similarly, starting from~$\tilde M^{-1}$ instead of~$\tilde M$ we get the relations \eqref{M-prop-b} for $\tilde M^{(1)}$ provided~$\alpha_1$ or~$\beta_1$ is $\neq 0$.

(ii)--(iii) Similar arguments apply at $k=\varkappa_2,\dots,\varkappa_6$.
\end{proof}

\section[How the solution $u(x,t)$ can be recovered from eigenfunctions]{How the solution $\boldsymbol{u(x,t)}$ can be recovered from eigenfunctions} \label{sec:3}

Starting from a solution $u(x,t)$ of the Novikov equation we have introduced
\begin{enumerate}\itemsep=0pt
\item[(a)]
a $3\times 3$ matrix-valued function $M(x,t;k)$, solution of the system of integral equations \eqref{M-int3},
\item[(b)]
a new variable $y(x,t)$ def\/ined by \eqref{y}.
\end{enumerate}
In this section, introducing $\hat M(y,t;k)\coloneqq M(x(y,t),t;k)$ through (a) and (b), we will show that $u(x,t)$ can be recovered in terms of $\hat M(y,t;k)$ evaluated at $k=\eul^{\frac{\ii\pi}{6}}$.

\subsection[A Lax pair appropriate for small $z$]{A Lax pair appropriate for small $\boldsymbol{z}$}

Coming back to the original Lax pair \eqref{Lax-vec}, let us introduce another Lax pair, whose solutions are well-controlled at $k=\kappa_l\coloneqq\eul^{\frac{\ii\pi}{6}+\frac{\ii\pi(l-1)}{3}}$, $l=1,\dots,6$; here $\{\kappa_l\}_{l=1}^6$ are characterized by the property $z(\kappa_l)=0$. Def\/ine $\tilde\Phi^{(0)}\equiv\tilde\Phi^{(0)}(x,t;k)$ (for $k\neq\kappa_1,\dots,\kappa_6$) by
\begin{gather*}
\tilde\Phi^{(0)}=P^{-1}(k)\Phi.
\end{gather*}
This reduces the Lax pair \eqref{Lax-vec} to the new one
\begin{subequations} \label{Lax-vec-0}
\begin{gather}\label{Lax-0}
\tilde\Phi^{(0)}_x -\Lambda(k)\tilde\Phi^{(0)}= \tilde U^{(0)}\tilde\Phi^{(0)},\\
\tilde\Phi^{(0)}_t - A(k)\tilde\Phi^{(0)} = \tilde V^{(0)}\tilde\Phi^{(0)},
\end{gather}
where
\begin{gather}
\tilde U^{(0)}(x,t;k)=P^{-1}(k)\left(U(x,t;k)-U_\infty(k)\right)P(k)\nonumber\\
\label{U-0}
{} =z^2(k)m(x,t)
\begin{pmatrix}
\frac{2\lambda_1(k)}{(3\lambda_1^2(k)-1)\lambda_1(k)} & \frac{\lambda_1(k)+\lambda_2(k)}{(3\lambda_1^2(k)-1)\lambda_1(k)} & \frac{\lambda_1(k)+\lambda_3(k)}{(3\lambda_1^2(k)-1)\lambda_1(k)}\vspace{1mm}\\
\frac{\lambda_2(k)+\lambda_1(k)}{(3\lambda_2^2(k)-1)\lambda_2(k)} & \frac{2\lambda_2(k)}{(3\lambda_2^2(k)-1)\lambda_2(k)} & \frac{\lambda_2(k)+\lambda_3(k)}{(3\lambda_2^2(k)-1)\lambda_2(k)}\vspace{1mm}\\
\frac{\lambda_3(k)+\lambda_1(k)}{(3\lambda_3^2(k)-1)\lambda_3(k)} & \frac{\lambda_3(k)+\lambda_2(k)}{(3\lambda_3^2(k)-1)\lambda_3(k)} & \frac{2\lambda_3(k)}{(3\lambda_3^2(k)-1)\lambda_3(k)}
\end{pmatrix},\\
\tilde V^{(0)}(x,t;k)=P^{-1}(k)\left(V(x,t;k)-V_\infty(k)\right)P(k)\nonumber\\
\label{V-0}
{} =P^{-1}(k)
\begin{pmatrix}
-(u+1)u_x & \frac{u_x}{z(k)}-(u^2+2u)\hat mz(k) & u_x^2\vspace{1mm}\\
\frac{u}{z(k)} & 0 & -\frac{u_x}{z(k)}-(u^2+2u)\hat mz(k)\vspace{1mm}\\
-u^2-2u & \frac{u}{z(k)} & (u+1)u_x
\end{pmatrix}P(k).
\end{gather}
\end{subequations}

\begin{Remark}
Notice that $\tilde U^{(0)}(x,t;k)$ has a f\/inite limit at $k=\kappa_l$, $l=1,\dots, 6$. Indeed, $\lambda_j(\kappa_l)=0$ for one, and only one, value of $j=1$, $2$ or $3$. For example, for $k=\kappa_1=\eul^{\frac{\ii\pi}{6}}$ we have $\lambda_2(\kappa_1)=0$ whereas $\lambda_1(\kappa_1)=-1$ and $\lambda_3(\kappa_1)=1$. Moreover, $\lim_{k\to\kappa_1}\frac{z^2(k)}{\lambda_2(k)}=\lambda_1(\kappa_1)\lambda_3(\kappa_1)=-1$ and thus{\samepage
\begin{gather*}
\lim_{k\to\kappa_1}\tilde U^{(0)}(x,t;k)=m(x,t)
\begin{pmatrix}
0 & 0 & 0 \\
-1 & 0 & 1 \\
0 & 0 & 0
\end{pmatrix}\not\equiv 0.
\end{gather*}
Similarly if $k\to\kappa_l$, $l=2,\dots,6$.}

We emphasize that this is dif\/ferent from the case of the DP equation, where, on the contrary, the equality $\tilde U^{(0)}(x,t;\kappa_1)=0$ holds for all~$x$ and~$t$.
\end{Remark}

\subsection{Eigenfunctions}

\subsubsection[Eigenfunctions appropriate for small $z$]{Eigenfunctions appropriate for small $\boldsymbol{z}$}

We introduce $M^{(0)}\equiv M^{(0)}(x,t;k)$ by
\begin{gather*}
M^{(0)}=\tilde\Phi^{(0)}\eul^{-x\Lambda-tA},
\end{gather*}
and determine $M^{(0)}$ as the solution of a system of integral equations similar to the system \eqref{M-int3} determining~$M$:
\begin{gather} \label{M0-int3}
M^{(0)}_{jl}(x,t;k)= I_{jl}+\int_{\infty_{j,l}}^x\eul^{-\lambda_j(k)(\xi-x)} \croch{(\tilde U^{(0)}M^{(0)})_{jl}(\xi,t;k)}\eul^{\lambda_l(k)
(\xi-x)}\dd\xi.
\end{gather}
In the case of the DP equation, $M^{(0)}(x,t;z(k))\big|_{z=0}\equiv I$. For the Novikov equation, this is not true; but, since $\tilde U^{(0)}(\xi,t;\kappa_1)\Delta\tilde U^{(0)}(x,t;\kappa_1)\equiv 0$ for any $\xi$ and $x$, and any diagonal matrix~$\Delta$, the solution of~\eqref{M0-int3} for $k=\kappa_1=\eul^{\frac{\ii\pi}{6}}$ can be written explicitly:
\begin{gather*} 
M_{jl}^{(0)}(x,t;\kappa_1)=
I+\int_{\infty_{jl}}^x\eul^{-\lambda_j(\kappa_1)(\xi-x)}\tilde U_{jl}^{(0)}(\xi,t;\kappa_1)\eul^{\lambda_l(\kappa_1)(\xi-x)}\dd\xi\equiv I_{jl}+L_{jl}(x,t),
\end{gather*}
where
\begin{gather*} 
L(x,t) =
\begin{pmatrix}
0 & 0 & 0 \\
\int_x^{\infty}m(\xi,t)\eul^{x-\xi}\dd\xi & 0 & \int_{-\infty}^x m(\xi,t)\eul^{\xi-x}\dd\xi\\
0 & 0 & 0
\end{pmatrix}.
\end{gather*}
Similarly for $k=\kappa_2,\dots,\kappa_6$. Using that $m=u-u_{xx}$, we see that the non-zero entries~$L_{21}$ and~$L_{23}$ reduce to
\begin{gather*}
L_{21}=u+u_x,\qquad L_{23}=u-u_x,
\end{gather*}
and thus $M^{(0)}\big(x,t;\eul^{\frac{\ii\pi}{6}}\big)$ can be explicitly expressed in terms of $u(x,t)$:
\begin{gather} \label{M0-1}
M^{(0)}\big(x,t;\eul^{\frac{\ii\pi}{6}}\big) =
\begin{pmatrix}
1 & 0 & 0 \\
u+u_x & 1 & u-u_x \\
0 & 0 & 1
\end{pmatrix} (x,t).
\end{gather}

\subsubsection{Comparison of eigenfunctions}

We will get the value of $M$ at $k=\eul^{\frac{\ii\pi}{6}}$ from that of $M^{(0)}$ by using that $M$ and $M^{(0)}$ are related. We indeed have
\begin{gather*}
M=P^{-1}D^{-1}\Phi\eul^{-y\Lambda-tA},\qquad
M^{(0)}=P^{-1}\Phi^{(0)}\eul^{-x\Lambda-tA},
\end{gather*}
where $\Phi$ and $\Phi^{(0)}$ are solutions of the same system of linear dif\/ferential equations \eqref{Lax-vec}. They are then related by $\Phi=\Phi^{(0)}C$ where $C\equiv C(k)$ is independent of $(x,t)$. Thus, $M$ and $M^{(0)}$ are related by
\begin{gather*}
M(x,t;k)=P^{-1}(k)D^{-1}(x,t)P(k)M^{(0)}(x,t;k)\eul^{x\Lambda(k)+tA(k)}C(k)\eul^{-y(x,t)\Lambda(k)-tA(k)}.
\end{gather*}
Now, since for $k\not\in\Sigma$, $M^{(0)}$ and $M$ are bounded as $x\to-\infty$ and have the same limit as $x\to+\infty$:
\begin{gather*}
M,\,M^{(0)}\xrightarrow[x\to+\infty]{}I,
\end{gather*}
it follows that $C(k)\equiv I$. Finally, we get
\begin{gather*}
M(x,t;k)=P^{-1}(k)D^{-1}(x,t)P(k)M^{(0)}(x,t;k)\eul^{(x-y(x,t))\Lambda(k)},
\end{gather*}
where $y$ is as in \eqref{y}:
\begin{gather*}
y(x,t)=x-\int_x^{\infty}\big(q^2(\xi,t)-1\big)\dd\xi.
\end{gather*}
In particular, at $k=\kappa_1\equiv\eul^{\frac{\ii\pi}{6}}$ we have
\begin{gather*}
P^{-1}(k)D^{-1}(x,t)P(k)\big|_{k=\eul^{\frac{\ii\pi}{6}}}=
\begin{pmatrix}
\frac{q^{-1}+q}{2} & 0 & \frac{q^{-1}-q}{2} \\
1-q^{-1} & 1 & 1-q^{-1} \\
\frac{q^{-1}-q}{2} & 0 & \frac{q^{-1}+q}{2}
\end{pmatrix}
\end{gather*}
and thus, using \eqref{M0-1},
\begin{gather} \label{M-1}
M\big(x,t;\eul^{\frac{\ii\pi}{6}}\big) = \begin{pmatrix}
\frac{q^{-1}+q}{2}\eul^{y-x} & 0 & \frac{q^{-1}-q}{2}\eul^{x-y} \\
(1-q^{-1}+u+u_x)\eul^{y-x} & 1 & (1-q^{-1}+u-u_x)\eul^{x-y} \\
\frac{q^{-1}-q}{2}\eul^{y-x} & 0 & \frac{q^{-1}+q}{2}\eul^{x-y}
\end{pmatrix}.
\end{gather}

\subsection[Recovering $u(x,t)$ from eigenfunctions]{Recovering $\boldsymbol{u(x,t)}$ from eigenfunctions}

\begin{Proposition}
Let $u(x,t)$ be the solution of the Cauchy problem~\eqref{N-Cauchy} for the Novikov equation. Let $\hat M(y,t;k)\coloneqq M(x(y,t),t;k)$ where $M(x,t;k)$ is the solution of the system~\eqref{M-int3} and $y(x,t)$ is defined by \eqref{y}, both built on $u(x,t)$.

Then we can recover $u(x,t)$ from values of the eigenfunctions $\hat M_{jl}(y,t;k)$ at $k=\eul^{\frac{\ii\pi}{6}}$. We indeed have the following parametric representation:
\begin{subequations} \label{u-x}
\begin{gather}\label{u-u}
u(x,t)=\hat u(y(x,t),t),
\end{gather}
where $x(y,t)$ and $\hat u(y,t)$ are given by
\begin{gather}\label{x-y}
x(y,t)=y+\frac{1}{2}\ln\frac{\hat M_{33}\big(y,t;\eul^{\frac{\ii\pi}{6}}\big)}{\hat M_{11}\big(y,t;\eul^{\frac{\ii\pi}{6}}\big)},\\
\hat u(y,t)=\frac{1}{2}\hat N_1(y,t)\left(\frac{\hat M_{33}\big(y,t;\eul^{\frac{\ii\pi}{6}}\big)}{\hat M_{11}(y,t;\eul^{\frac{\ii\pi}{6}})}\right)^{1/2}+\frac{1}{2}\hat N_3(y,t)\left(\frac{\hat M_{33}\big(y,t;\eul^{\frac{\ii\pi}{6}}\big)}{\hat M_{11}\big(y,t;\eul^{\frac{\ii\pi}{6}}\big)}\right)^{-1/2}-1,\label{u-y}
\end{gather}
where
\begin{gather}
\hat N_k(y,t)\coloneqq\sum_{j=1}^3\hat M_{jk}\big(y,t;\eul^{\frac{\ii\pi}{6}}\big), \qquad k=1,2,3.\label{N-y}
\end{gather}
\end{subequations}
\end{Proposition}

\begin{proof}
\eqref{x-y} follows from \eqref{M-1} and from the def\/inition \eqref{y} of $y(x,t)$. Further, multiplying \eqref{M-1} by the row vector
$(\begin{matrix}1 & 1 & 1\end{matrix})$ from the left and introducing
\begin{gather*}
N(x,t)\equiv\begin{pmatrix}N_1(x,t)&N_2(x,t)&N_3(x,t)\end{pmatrix}\coloneqq\begin{pmatrix}1&1&1\end{pmatrix}M\big(x,t;\eul^{\frac{\ii\pi}{6}}\big)
\end{gather*}
we have
\begin{gather*}
N_1(x,t)= (1+u(x,t)+u_x(x,t))\eul^{y-x},\\
N_2(x,t)=1,\\
N_3(x,t)=(1+u(x,t)-u_x(x,t))\eul^{x-y}
\end{gather*}
and thus $u(x,t)$ can be obtained in terms of $N(x,t)$ as follows:
\begin{gather}
u(x,t) = \frac{1}{2}N_1(x,t)\eul^{x-y} + \frac{1}{2}N_3(x,t)\eul^{y-x} - 1.
\label{N-x2}
\end{gather}
Now notice that \eqref{x-y} reads
\begin{gather}
\eul^{2(x-y)} = \frac{\hat M_{33}\big(y,t;\eul^{\frac{\ii\pi}{6}}\big)}
{\hat M_{11}\big(y,t;\eul^{\frac{\ii\pi}{6}}\big)}.\label{e-xy}
\end{gather}
Finally, introducing $\hat N_k$ as in \eqref{N-y}, using~\eqref{e-xy}, and writing~\eqref{N-x2} in the variables~$(y,t)$ (so that $N(x,t)=\hat N(y,t)$), we arrive at~\eqref{u-y}.
\end{proof}

\begin{Remark}
From \eqref{y} and \eqref{N-q}, another expression for $\hat u(y,t)$ follows:
\begin{gather}
(\hat u(y,t)+1)^2=\partial_t\,x(y,t)+1.
\end{gather}
Notice that the analogous expression in the case of the DP equation looks dif\/ferently~\cite{BS13}:
\begin{gather*}
\hat u_{\DP}(y,t)=\partial_t\,x(y,t)+1.
\end{gather*}
\end{Remark}

\begin{Remark}The structure of~\eqref{u-x} coincides with that for the multisoliton solution in \cite{M13} (see formulas~(3.3a) and~(3.3b) in~\cite{M13}), taking into account the relationship between solutions of the Novikov equation~\eqref{N1} on a non-zero constant background and solutions of \eqref{N5} on the zero background, presented in the Introduction.
\end{Remark}

\section{Riemann--Hilbert problem}

In Section~\ref{sec:3} we have shown how we can express the solution $u$ of the Cauchy problem for the Novikov equation by evaluating certain eigenfunctions -- solutions of the Lax pair equations. Notice they were def\/ined using the solution $u$ itself.

In the framework of the Riemann--Hilbert approach to the Cauchy problem for an integrable nonlinear equation, one is looking for obtaining these eigenfunctions in terms of the solution of an appropriate factorization problem, of Riemann--Hilbert type. The factorization problem is formulated in the complex plane of a spectral parameter $k$ whereas $x$ and $t$ play the role of parameters, and the data for this problem are uniquely determined, in spectral terms, by the initial data for the Cauchy problem.

\subsection[RH problem satisf\/ied by $\hat M$]{RH problem satisf\/ied by $\boldsymbol{\hat M}$}\label{RH-hatM}

Let $M(x,t;k)$ be as in Section \ref{sec:3}, solution of the system of integral equations~\eqref{M-int3}. The key observation is that the limiting values $M_{\pm}(x,t;k)$ (on $l_{\nu}$) of $M(x,t;k')$ as $k'\to k$ from the positive or negative side of~$l_{\nu}$, $\nu=1,\dots,6$ are related as follows:
\begin{gather}\label{M-scat}
M_+(x,t;k)=M_-(x,t;k)\eul^{Q(x,t;k)}S_0(k)\eul^{-Q(x,t;k)}, \qquad k\in l_{\nu}.
\end{gather}
We indeed have $M_{\pm}=\hat\Phi_{\pm}\eul^{-Q}$ where $\hat\Phi_{\pm}$ are two solutions of the system of ordinary dif\/ferential equations~\eqref{Lax-vec-2}. They are then related by $\hat\Phi_+=\hat\Phi_-S_0$ where $S_0$ is a matrix independent of~$(x,t)$. Considering \eqref{M-scat} at $t=0$ we see that $S_0(k)$ is completely determined by $u(x,0)$, i.e., by the initial data for the Cauchy problem \eqref{N-Cauchy}, via the solution $M(x,0;k)$ of the system~\eqref{M-int3} whose coef\/f\/icients are determined by~$u(x,0)$:
\begin{gather*}
S_0(k)=\eul^{-Q(x,0;k)}M_-^{-1}(x,0;k)M_+(0,0;k)\eul^{Q(x,0;k)}
\end{gather*}
(recall that $Q(x,0;k)=\left(x-\int_x^\infty(q^2(\xi,0)-1)\dd\xi\right)\Lambda(k)$ and thus $Q(x,0;k)$ is also determined by~$u(x,0)$).

Moreover, $S_0(k)$ has a special matrix structure: for $k\in l_1\cup l_4\equiv\D{R}$,
\begin{gather}\label{S-SS-1}
S_0(k)=
\begin{pmatrix}
1 & 0 & 0\\
-r(k) & 1 & 0\\
0 & 0 & 1
\end{pmatrix}
\begin{pmatrix}
1 & \overline{r(k)} & 0\\
0 & 1 & 0\\
0 & 0 & 1
\end{pmatrix},
\end{gather}
where $r(k)\in L^\infty(\D{R})$ and $r(k)=\ord(k^{-1})$ as $k\to\pm\infty$ (this structure follows from the analysis of the behavior of
$M_\pm(x,0;k)$ as $x\to\pm\infty$; for the details, see \cite[Section~3.1]{BS13}), whereas the expression of~$S_0(k)$ for the other parts of~$\Sigma$ follows from the symmetries stated in Proposition~\ref{p2}. Thus the jump matrices on all parts of the contour are determined in terms of a single scalar function, the \emph{reflection coefficient}~$r(k)$.

On the other hand, in order to have explicit exponentials in the r.h.s.\ of~\eqref{M-scat}, one replaces the pair of parameters $(x,t)$ by $(y,t)$, where $y=y(x,t)$ was introduced in~\eqref{y}, see also \eqref{Q}. Following this observation, \eqref{M-scat} can be rewritten as
\begin{gather}\label{jump-y}
\hat M_+(y,t;k)=\hat M_-(y,t;k)\eul^{y\Lambda(k)+tA(k)}S_0(k) \eul^{-y\Lambda(k)-tA(k)},\qquad k\in\Sigma
\end{gather}
where $\Sigma=l_1\cup\dots\cup l_6$ and
\begin{gather}\label{hat-M}
M(x,t;k)=\hat M(y(x,t),t;k).
\end{gather}

Now \eqref{jump-y} says that $\hat M$ -- as in \eqref{hat-M} -- is solution of a factorization problem:

\begin{fact-pb*} Given $S_0(k)$, $k\in\Sigma$, f\/ind a piece-wise (w.r.t.\ $\Sigma$) meromorphic (in $k$), $3\times 3$ matrix-valued function $\hat M(y,t;k)$ ($y$ and $t$ are parameters), whose limiting values satisfy the jump condition \eqref{jump-y}.
\end{fact-pb*}

This condition, being supplemented by conditions at possible poles of $\hat M(y,t;k)$, by a normali\-za\-tion condition, and by certain structural conditions, will constitute the RH problem satisf\/ied by~$\hat M$.

\begin{add-conds*} The dependence on $k$ of the eigenvalues $\lambda_j(k)$ in \eqref{M-int3} is exactly the same as in the case of the DP equation, which implies that most conditions involved in the RH problem for the Novikov equation have exactly the same form as for the Degasperis--Procesi equation, cf.~\cite{BS13}.
\begin{enumerate}\itemsep=0pt
\item For the residue conditions at poles (if any) $k_n\in\D{C}\setminus\Sigma$ of $\hat M(y,t;k)$, see \cite[Section~3.2]{BS13} and \cite[Section~2.3.5]{BS15}. Recall the following result (\cite{BC}): Generically, there are at most a~f\/inite number of poles $k_n$ lying in $\D{C}\setminus\Sigma$, each of them being simple, with residue conditions of a special matrix form. In particular,
\begin{gather*}
\Res_{k=k_n}\hat M(y,t;k)=\lim_{k\to k_n}\hat M(y,t;k)\eul^{y\Lambda(k)+tA(k)}v_n\eul^{-y\Lambda(k)-tA(k)},
\end{gather*}
where $v_n$ is some constant $3\times 3$ matrix with only one non-zero entry at a position depending on the sector of $\D{C}\setminus\Sigma$ to which $k_n$ belongs.
\item The normalization condition $\hat M(y,t;\infty)=I$ is Proposition~\ref{prop:p1}(ii).
\item The structural conditions on the polar parts of $\hat M$ at $\varkappa_1,\dots,\varkappa_6$ have exactly the same form as for the DP equation, cf.~\cite{BS13}.
\item On the other hand, the structural conditions at $k=\kappa_1,\dots,\kappa_6$, see~\eqref{M0-1} and~\eqref{M-1}, are dif\/ferent from those for the DP equation (cf.~(3.12) and~(3.13b) in~\cite{BS13}).
\end{enumerate}
\end{add-conds*}

Summarizing, we get that the eigenfunctions $\hat M(y,t;k)$ -- def\/ined as in Section \ref{sec:3} -- satisfy the following Riemann--Hilbert problem:

\begin{rh-pb*}[for the Novikov equation]
Given $r(k)\in L^\infty(\D{R})$ such that $r(k)=\ord(k^{-1})$ as $k\to\pm\infty$ and $\accol{k_n,v_n}_{n=1}^N$, f\/ind a piece-wise (w.r.t.\ $\Sigma=l_1\cup\dots\cup l_6$) meromorphic (in the complex $k$-plane), $3\times 3$ matrix-valued function $\hat M(y,t;k)$ satisfying the following conditions:
\begin{enumerate}\itemsep=0pt
\item[(1)]\emph{Jump conditions}:
\begin{gather*}
\hat M_+(y,t;k)=\hat M_-(y,t;k)S(y,t;k),\qquad k\in\Sigma\setminus\accol{\varkappa_1,\dots,\varkappa_6},
\end{gather*}
where $\varkappa_{\nu}\coloneqq\eul^{\frac{\ii\pi}{3}(\nu-1)}$ and
\begin{gather*}
S(y,t;k)=\eul^{y\Lambda(k)+tA(k)}S_0(k)\eul^{-y\Lambda(k)-tA(k)},\end{gather*}
where $S_0(k)$ is given in terms of $r(k)$ by \eqref{S-SS-1} for $k\in l_1\cup l_4$ and then by the respective symmetries for the other parts of $\Sigma$. Here $\Lambda(k)$ and $A(k)$ are diagonal $3\times 3$ matrices with diagonal entries $\lambda_i(k)$ and $A_i(k)$, respectively, see \eqref{La} and \eqref{Aa}, where we have performed the change of parameter $z=z(k)$, see~\eqref{z-k} and~\eqref{la-k}.
\item[(2)] \emph{Residue conditions} (generic): For some $k_n\in\D{C}\setminus\Sigma$, $n=1,\dots,N$,
\begin{gather}\label{res}
\Res_{k=k_n}\hat M(y,t;k)=\lim_{k\to k_n}\hat M(y,t;k)\eul^{y\Lambda(k)+tA(k)}v_n\eul^{-y\Lambda(k)-tA(k)},
\end{gather}
where $v_n$ is some constant $3\times 3$ matrix with only one non-zero entry, whose position depends on the sector of $\D{C}\setminus\Sigma$ containing $k_n$ (dictated by the order of $\Re\lambda_j(k)$ in the sector, see~\cite{BC}). For example, if $k_n\in\Omega_1$, the non-zero entry of $v_n$ can be either~$(v_n)_{12}$ or~$(v_n)_{23}$. Then the positions (as well as the values) of the non-zero entries of~$v_n$ in the other sectors of~$\D{C}\setminus\Sigma$ are determined by the symmetries (S1)--(S3) from Proposition~\ref{p2}.
\item[(3)] \emph{Normalization condition}:
\begin{gather*}
\hat M(y,t;k)=I+\osmall(1)\qquad \text{as}\quad k\to\infty.
\end{gather*}
\item[(4)] \emph{Boundedness condition}:
$\hat M(y,t;k)$ and $\hat M^{-1}(y,t;k)$ are bounded in the closure of each component of $\D{C}\setminus\Sigma$ apart from vicinities of $k_1,\dots,k_N$ and $\varkappa_1,\dots,\varkappa_6$.
\item[(5)] \emph{Structural conditions} at $k=\varkappa_1,\dots,\varkappa_6$, where $\varkappa_{\nu}\coloneqq\eul^{\frac{\ii\pi}{3}(\nu-1)}$:
\begin{enumerate}\itemsep=0pt
\item[i)] The singular behavior of $\hat M$ and $\hat M^{-1}$ at $k=\varkappa_1=1$ is described by the structural conditions~\eqref{singul},~\eqref{M-prop} (where $\alpha$, $\beta$, $\alpha_1$, $\beta_1,\dots$ are not specif\/ied).
\item[ii)] At $k=\varkappa_2,\dots,\varkappa_6$ we have similar structural conditions on the singularities of~$\hat M$ and~$\hat M^{-1}$ (see Proposition~\ref{p3}, (ii)--(iii)).
\end{enumerate}
\item[(6)] \emph{Structural conditions} at $k=\kappa_1,\dots,\kappa_6$, where $\kappa_l\coloneqq\eul^{\frac{\ii\pi}{6}+\frac{\ii\pi}{3}(l-1)}$:
\begin{enumerate}\itemsep=0pt
\item[i)]
For $k=\kappa_1=\eul^{\frac{\ii\pi}{6}}$,
\begin{gather} \label{M-struc}
\hat M\big(y,t;\eul^{\frac{\ii\pi}{6}}\big)=\begin{pmatrix}
\frac{\hat q^{-1}+\hat q}{2} & 0 & \frac{\hat q^{-1}-\hat q}{2} \\
p_1 & 1 & p_2 \\
\frac{\hat q^{-1}-\hat q}{2} & 0 & \frac{\hat q^{-1}+\hat q}{2}
\end{pmatrix}
\begin{pmatrix}
f & 0 & 0 \\
0 & 1 & 0 \\
0 & 0 & f^{-1}
\end{pmatrix}
\end{gather}
with some (unspecif\/ied) $\hat q(y,t)>0$, $f(y,t)>0$, $p_1(y,t)\in\D{R}$ and $p_2(y,t)\in\D{R}$ (see~\eqref{M-1}).
\item[ii)]
For $k=\kappa_2,\dots,\kappa_6$, the corresponding structure of $\hat M(y,t;\kappa_l)$ follows from~\eqref{M-struc} taking into account the symmetries of Proposition~\ref{p2}.
\end{enumerate}
\end{enumerate}
\end{rh-pb*}

Now we observe that this RH problem can be def\/ined using the initial data of the Cauchy problem.

\begin{properties*}[of the main RH problem]
Let $u_0(x)$ be given satisfying the assumptions made for the Cauchy problem \eqref{N-Cauchy}.
\begin{enumerate}\itemsep=0pt
\item[i)]
$u_0(x)$ def\/ines the main RH problem as follows.
\begin{enumerate}\itemsep=0pt
\item[ia)] We f\/irst def\/ine the jump matrix $S_0(k)$ by
\begin{gather*}
S_0(k)=\eul^{-Q(0,0;k)}M_-^{-1}(0,0;k)M_+(0,0;k)\eul^{Q(0,0;k)},
\end{gather*}
where $M(x,0;k)$ and $Q(x,0;k)$ can be completely determined by $u_0(x)$. Notice that $S_0(k)$ has necessarily the structure \eqref{S-SS-1} with some $r(k)$.
\item[ib)] The $k_n$'s in $\D{C}\setminus\Sigma$ and the $v_n$'s in \eqref{res} are determined by $u_0(x)$.
\end{enumerate}
\item[ii)] This RH problem has a solution.
\item[iii)] The solution is unique.
\end{enumerate}
\end{properties*}

\begin{proof}
ia) The expression giving $S_0(k)$ derives from \eqref{M-scat}. Further, $M(x,0;k)$ is the solution of the system of integral equations~\eqref{M-int3} written for $t=0$, and this system is completely determined by $u_0(x)$ through $\hat m_0(x)\coloneqq u_0(x)-u_0''(x)+1$. By \eqref{q} we indeed have $q^2(x,0)=\hat m_0^{2/3}(x)$, and in the def\/inition of $\hat U(x,0;k)$ (see~\eqref{hat-U} and~\eqref{U-tilde}) the term $\hat m(x,0)$ is also $\hat m_0(x)$. Similarly for $Q(x,0;k)=y(x,0)\Lambda(k)$ where $y(x,0)$ is determined by $q^2(x,0)$.

ib) The $k_n$'s are determined by $M(x,0;k)$ and $v_n$ is determined by~\eqref{res} considered at $t=0$, hence also by $u_0(x)$.

ii) Under the assumptions made on $u_0$ there exists a solution $u(x,t)$ of the Cauchy prob\-lem~\eqref{N-Cauchy}. The main RH problem has then as solution the associated eigenfunction $\hat M(y,t;k)$.

iii) The normalization condition supplemented by the structural conditions~\eqref{M-struc} and~\eqref{singul}, \eqref{M-prop} provides uniqueness of the solution (see \cite[Proposition~3.2]{BS-JGA}).
\end{proof}

\subsection{Main theorem}\label{sec:main-thm}

Starting directly from the main RH problem we show how its solution gives a representation of the solution of the Cauchy problem for the Novikov equation.

\begin{Theorem}\label{thm:main}
Let $u_0(x)$ satisfy assumptions made for the Cauchy problem \eqref{N-Cauchy} for the Novikov equation. Let $\hat M(y,t;k)$ be the solution of the associated main RH problem {\rm (1)--(6)}. By evaluating~$\hat M$ at $k=\eul^{\frac{\ii\pi}{6}}$ we get the parametric representation~\eqref{u-x} for the solution $u(x,t)$ of the Cauchy problem~\eqref{N-Cauchy}:
\begin{gather*}
u(x,t)=\hat u(y(x,t),t),
\end{gather*}
where $x(y,t)$ and $\hat u(y,t)$ are given by
\begin{gather*}
x(y,t)=y+\frac{1}{2}\ln\frac{\hat M_{33}\big(y,t;\eul^{\frac{\ii\pi}{6}}\big)}{\hat M_{11}\big(y,t;\eul^{\frac{\ii\pi}{6}}\big)},\\
\hat u(y,t)=\frac{1}{2}\hat N_1(y,t)\left(\frac{\hat M_{33}\big(y,t;\eul^{\frac{\ii\pi}{6}}\big)}{\hat M_{11}\big(y,t;\eul^{\frac{\ii\pi}{6}}\big)}\right)^{1/2}+\frac{1}{2}\hat N_3(y,t)\left(\frac{\hat M_{33}\big(y,t;\eul^{\frac{\ii\pi}{6}}\big)}{\hat M_{11}\big(y,t;\eul^{\frac{\ii\pi}{6}}\big)}\right)^{-1/2}-1,
\end{gather*}
with $\hat N_k(y,t):=\sum\limits_{j=1}^3\hat M_{jk}\big(y,t;\eul^{\frac{\ii\pi}{6}}\big)$, $k=1,3$.
\end{Theorem}

\begin{Remark}
The solution of the DP equation is extracted from the solution of the associated RH problem in a dif\/ferent way (see \cite[equation~(3.15)]{BS13}).
\end{Remark}

\subsection{Applications}\label{sec:applis}

The representation \eqref{u-x} of the solution of the Cauchy problem \eqref{N-Cauchy} for the Novikov equa\-tion~\eqref{N5} in terms of the solution of the associated main RH problem (1)--(6) (given in Theorem~\ref{thm:main}) can be used in principle, as for other integrable peakon equations (see, e.g., \cite{BIS10,BKST,BLS15,BS06,BS-JGA,BS-D,BS08,BS-F,BS13,BS15,BSZ11}), to construct soliton-type solutions and also to analyse the long time behavior of the solution.

\subsubsection{Soliton-type solutions}

In the case $r(k)\equiv 0$, solving the RH problem reduces to solving a system of linear algebraic equations associated with the residue conditions. Following this way, one can describe not only smooth classical solutions, but also singular and/or multivalued solutions associated with breaks of bijectivity in the change of variables $x\mapsto y$ (see, e.g.,~\cite{BS15,BSZ11}).

\subsubsection{Long-time asymptotics}

Noticing that the form and the $k$-dependence of the jump matrix, see~\eqref{jump-y} and~\eqref{S-SS-1}, are the same as in the case of the DP equation, a reasonable conjecture is that the long-time analysis for the Novikov equation can be done following the same steps as for the DP equation~\cite{BS13}, by adapting the nonlinear steepest descent method by Deift and Zhou~\cite{DZ}.

But there is a problem that prevents from transferring the long-time analysis literally from the case of the DP equation to the Novikov equation. Indeed, in \cite{BLS15, BS13} the long-time analysis was actually done for a matrix RH problem regular at the points $\varkappa_1,\dots,\varkappa_6$ whereas the original RH problem, as in the case of the Novikov equation, was singular at these points. This discrepancy disappeared in \cite{BLS15,BS13} after left multiplying by the row vector $(\begin{matrix} 1 & 1 & 1\end{matrix})$ and thus passing from the matrix to a (regular) row vector RH problem and exploiting the uniqueness of the solution of the latter problem. In turn, the solution of the row vector RH problem evaluated at $\kappa_1,\dots,\kappa_6$, and particularly, at $k=\kappa_1=\eul^{\frac{\ii\pi}{6}}$, gave, in a simple way, the quantity $\eul^{x(y,t)-y}$, from which it was straightforward to deduce a parametric representation for $u(x,t)$, see~\eqref{u-x}.

Notice that a similar situation takes place for other Camassa--Holm-type equations including the Camassa--Holm equation itself, see \cite{BS06,BS08}. On the contrary, in the case of the Novikov equation, the situation is dif\/ferent. Indeed, the solution of the similar row vector RH problem cannot be directly used to recover $\eul^{x(y,t)-y}$ since, being evaluated at $k=\eul^{\frac{\ii\pi}{6}}$, it has the form
\begin{gather*}
\begin{pmatrix}(1+u+u_x)\eul^{y-x}&1&(1+u-u_x)\eul^{x-y}\end{pmatrix},
\end{gather*}
left multiplying~\eqref{M-1} by the row vector $(\begin{matrix} 1 & 1 & 1\end{matrix})$. This issue has forced us to use instead the solution of the matrix RH problem~\eqref{M-1} in~\eqref{u-x}. But now there is an (open) problem of relating the solution of the original (singular) problem to the solution of the limiting (as $t\to\infty$) regular RH problem. In~\cite{BLS15, BS13}, it was the latter problem that was used for obtaining the main asymptotic term of $u(x,t)$ in an explicit form, with parameters determined by the initial data in terms of the associated ref\/lection coef\/f\/icient.

\subsubsection{Initial boundary value problem}

Eigenfunctions associated with the Lax pair equations~\eqref{M} and \eqref{Lax-vec-0} via integral Fredholm equations of type~\eqref{M-int} with an appropriate choice of $(x^*,t^*)$ allows formulating a RH problem suitable for analyzing initial boundary value (or half-line) problems following the procedure presented in~\cite{BLS15} in the case of the DP equation.

\pdfbookmark[1]{References}{ref}
\LastPageEnding

\end{document}